\documentclass[conference, 10pt]{IEEEtran}
\IEEEoverridecommandlockouts
\usepackage{amsmath}
\usepackage{cite}
\usepackage{amsmath,amssymb,amsfonts}
\usepackage[ruled,linesnumbered,lined]{algorithm2e}
\usepackage{graphicx}
\usepackage{textcomp}
\usepackage{xcolor}
\usepackage{bm}
\usepackage{booktabs}
\usepackage{caption,subcaption}

\usepackage{amsthm,amssymb,amsfonts}

\newtheorem{theorem}{Theorem}

\SetKwInput{Input}{input}
\SetKwInput{Output}{output}

\usepackage{array}
\newcolumntype{L}[1]{>{\raggedright\let\newline\\\arraybackslash\hspace{0pt}}m{#1}}
\newcolumntype{C}[1]{>{\centering\let\newline\\\arraybackslash\hspace{0pt}}m{#1}}
\newcolumntype{R}[1]{>{\raggedleft\let\newline\\\arraybackslash\hspace{0pt}}m{#1}}

\def\BibTeX{{\rm B\kern-.05em{\sc i\kern-.025em b}\kern-.08em
		T\kern-.1667em\lower.7ex\hbox{E}\kern-.125emX}}

\begin{document}
	\title{Two-Layer Microwave Linear Analog Computer (MiLAC)-aided Multi-user MISO Networks}
	\author{
		\IEEEauthorblockN{Xiaohua Zhou$^{*}$,  Tianyu Fang$^{\dagger}$,
			Yijie Mao$^{*}$, Bruno Clerckx$^{\ddagger}$ }
		\IEEEauthorblockA{
			$^{*}$School of Information Science and Technology, ShanghaiTech University, Shanghai 201210, China \\
			$^{\dagger}$Centre for Wireless Communications, University of Oulu, Finland \\
		$^{\ddagger}$Department of Electrical and Electronic Engineering, Imperial College London, London SW7 2AZ, U.K.\\
			Email:\{zhouxh3, 
			maoyj\}@shanghaitech.edu.cn, tianyu.fang@oulu.fi, b.clerckx@imperial.ac.uk
		}
			\thanks{This work has been supported by the National Nature Science Foundation of China under Grant 62571331.  \textit{(Corresponding author: Yijie Mao)}}
		\\[-2.5 ex]
		\vspace{-1.1cm}
	}
	\maketitle
	\vspace{-1.6cm}
	\thispagestyle{empty}
	\pagestyle{empty}
	\begin{abstract}
		Microwave linear analog computer (MiLAC)-aided transmit beamforming, which processes transmitted symbols entirely in the analog domain, has recently emerged as a promising alternative to fully digital or hybrid beamforming architectures for single-user multi-antenna systems. However, recent studies have shown that deploying a single lossless and reciprocal MiLAC at the transmitter cannot achieve the same capacity as fully digital beamforming in multi-user scenarios. To address this limitation, we propose a novel two-layer MiLAC-aided beamforming architecture at the transmitter for a downlink multi-user multiple-input single-output (MISO) network. Leveraging microwave network theory, we first prove that lossless and reciprocal two-layer MiLAC-aided beamforming can achieve the same performance as digital beamforming, and we derive a closed-form mapping from digital beamforming to two-layer MiLAC analog beamforming. Furthermore, we formulate a sum-rate maximization problem and develop an efficient optimization framework to jointly optimize the power allocation and the scattering matrices for the proposed two-layer MiLAC architecture. Numerical results validate our theoretical findings and demonstrate that two-layer MiLAC achieves the same sum-rate performance as fully digital beamforming.
	\end{abstract}
	
	\begin{IEEEkeywords}
		Microwave linear analog computer (MiLAC), multi-user multiple-input single-output (MISO), sum-rate maximization.
	\end{IEEEkeywords}
	\vspace{-0.7cm}
	\section{Introduction}
	Beyond-diagonal reconfigurable intelligent surfaces (BD-RIS) leverage multiport network theory to model inter-element coupling and reconfigure the wireless environment \cite{BDRIStutorial}. Building on this principle, the microwave linear analog computer (MiLAC) employs a reconfigurable multiport microwave network at the transmitter to perform analog beamforming prior to radiation \cite{nerini2025milac1,nerini2025milac2}. In MiLAC, baseband signals are upconverted to radio-frequency (RF) and processed through the network, where analog transformations are realized via microwave interactions before direct radiation. This intrinsic analog computing capability eliminates the need for digital precoding, enabling a compact and energy-efficient solution for extremely large-scale multiple-input multiple-output (MIMO) systems, with reduced RF chain requirements, relaxed analog-to-digital converter (ADC)/digital-to-analog converter (DAC) resolution, and significantly lower computational complexity.
	
	The study of MiLAC is still in its early stages. A recent study \cite{nerini2025milac-capacity} has shown that fully-connected, lossless, and reciprocal one-layer MiLAC-aided beamforming achieves the same capacity as digital beamforming in single-user MIMO systems. However, such fully-connected architectures require interconnecting all ports, providing maximum flexibility at the cost of high circuit complexity and control overhead. To alleviate this issue, \cite{nerini2025mimo} proposed one-layer stem-connected MiLAC architectures, which significantly reduce hardware complexity while preserving capacity-achieving performance in single-user MIMO systems. However, when extended to multi-user scenarios, existing works \cite{multiuser-milac,limits-milac} have shown that lossless reciprocal one-layer MiLAC suffers from performance degradation compared to fully digital beamforming. To address this limitation, a hybrid digital–MiLAC beamforming architecture has been proposed, requiring only a small number of RF chains to interface a low-dimensional digital beamformer with a single-layer MiLAC-aided analog beamformer \cite{limits-milac}. While this approach reduces the number of required RF chains and preserves the performance of fully digital beamforming, the inclusion of a digital beamformer introduces the need for high-resolution RF chains. To the best of our knowledge, the design of a fully analog beamforming architecture capable of achieving digital beamforming performance in multi-user systems remains an open research problem.
	
	In this work, we address this challenge by proposing a novel two-layer MiLAC-aided beamforming architecture that achieves the same sum-rate performance as digital beamforming in a $K$-user multiple-input single-output (MISO) systems. This architecture employs two cascaded MiLACs at the transmitter and requires only $K$ RF chains for low-resolution DACs and $K$ amplifiers for power allocation (see Fig. \ref{Fig:sys-model}). Based on this framework, we analytically prove that the proposed architecture can achieve performance equivalent to fully digital beamforming in a general multi-user MISO system, and we derive a closed-form mapping from digital beamforming to two-layer MiLAC analog beamforming. Furthermore, we formulate a sum-rate maximization problem and develop an efficient optimization framework to jointly optimize the power allocation and the scattering matrix for the proposed two-layer MiLAC architecture. Numerical results validate the theoretical analysis and demonstrate the effectiveness of the proposed approach.

\par \textit{Notations:}  $ (\cdot)^\mathsf{T} $, $ (\cdot)^\mathsf{H} $, $ (\cdot)^* $, and $ \operatorname{tr}(\cdot) $ represent the transpose, Hermitian, complex conjugate, and trace, respectively.
	\vspace{-0.3cm}
	
	\section{System  Model}\label{Sec:system model}
		\vspace{-0.2cm}
	Consider a multi-user downlink communication network, where a base station (BS) equipped with $L$ transmit antennas simultaneously serves $K$ single-antenna users indexed by $\mathcal{K} = \{1, \cdots, K\}$. Unlike the MiLAC-assisted multi-user MISO network proposed in \cite{multiuser-milac,limits-milac}, where the transmitter employs a single-layer MiLAC to enable analog beamforming, this work introduces a novel two-layer MiLAC-aided transmitter. In the following, we elaborate on the transmission model and the MiLAC modeling for the proposed two-layer MiLAC-aided beamforming architecture.

	
	\subsection{Proposed Two-layer MiLAC-aided Transmission Model}
	As illustrated in Fig. \ref{Fig:sys-model}, the proposed transmitter architecture operates as follows. First, $K$ parallel symbol streams intended for the users are processed by $K$ RF chains, which convert the signals from the digital domain to the analog domain. These RF chains are connected to a $2K$-port MiLAC, referred to as MiLAC~1, which performs first-stage precoding based on the analog beamforming matrix $\mathbf{F} \in \mathbb{C}^{K \times K}$. Subsequently, $K$ amplifiers  are dedicated to power allocation. Finally, a $(K+L)$-port MiLAC, denoted as MiLAC~2, performs second-stage precoding before transmission based on the analog beamforming matrix $\mathbf{W} \in \mathbb{C}^{L \times K}$.

	\begin{figure}[!t]
		\centering
		\includegraphics[width=0.48\textwidth]{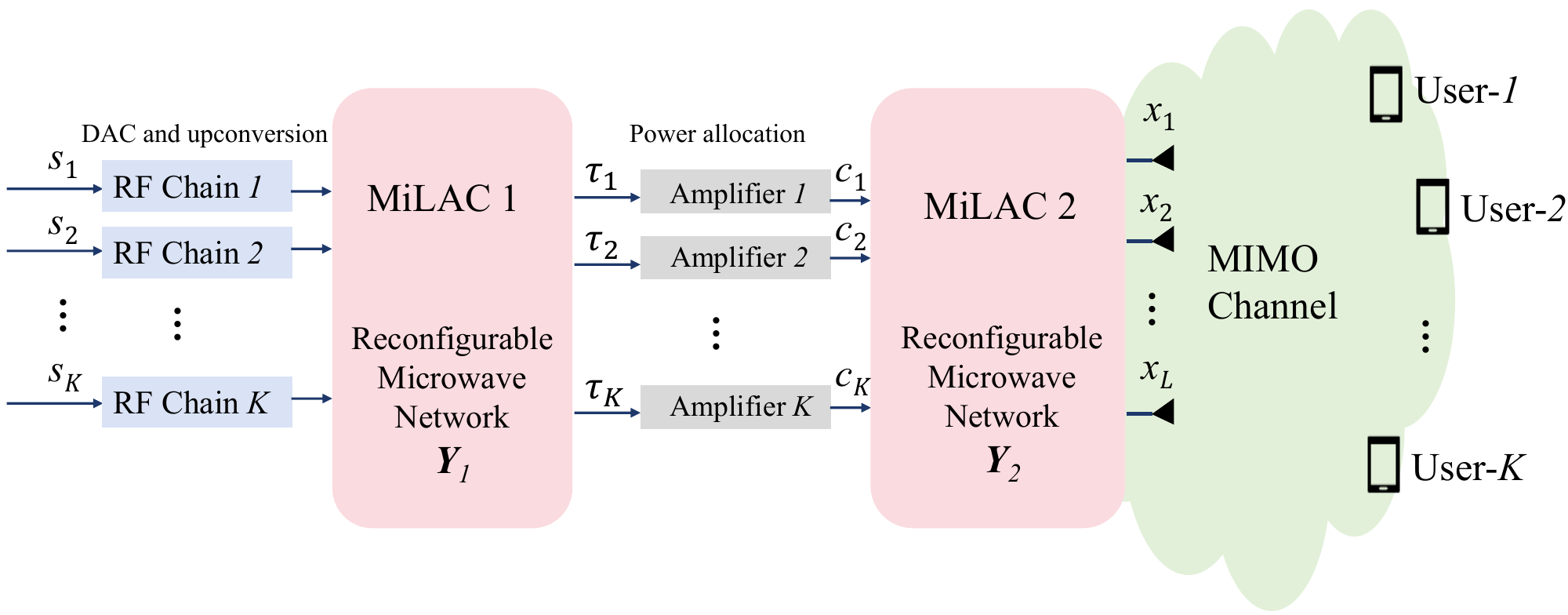}
		\caption{Block diagram of a multi-user MISO system with two-layer MiLAC architecture at the BS. }
		\label{Fig:sys-model}
		\vspace{-0.4cm}
	\end{figure} 
	
	
	Let $s_k \in \mathbb{C}$ denote the complex-valued data symbol intended for user $k$, satisfying $\mathbb{E}[|s_k|^2] = 1$. Define $\mathbf{s} = [s_1, \cdots, s_K]^T \in \mathbb{C}^{K \times 1}$ as the vector collecting the data symbols for all users. By adopting the proposed two-layer MiLAC-aided beamforming architecture at the transmitter, the resulting transmit signal at the output of MiLAC 1 is $ \bm \tau = \mathbf{F} \mathbf{s} $. Let $ \mathbf{P} \in \mathbb{C}^{K \times K}$ denote the power allocation matrix at the amplifier between MiLAC 1 and MiLAC 2. After power allocation by the amplifiers, the signal at the input of MiLAC 2 becomes $\mathbf{c}=\mathbf{P}^{\frac{1}{2}} \bm \tau$,
	where $ \mathbf{P}^{\frac{1}{2}} \in \mathbb{C}^{K \times K}$ is the square root of the power allocation matrix, i.e.,  $ \mathbf{P}^{\frac{1}{2}} = \operatorname{diag} (\sqrt{p_1}, \cdots, \sqrt{p_K}) $ with $ p_k $ being the power allocation for the $ k $-th symbol and $ \sum_{k=1}^K p_k \leq P_t $. Here,  $P_t$ is the total transmit power. After the whole two-layer MiLAC-aided beamforming architecture, transmit signal at the BS is 
	\begin{equation}
	\mathbf{x} = \mathbf{W}\mathbf{P}^{\frac{1}{2}} \mathbf{F} \mathbf{s}.
	\end{equation}
	 Note that the analog beamforming matrix employed by MiLAC must satisfy the underlying physical constraints of the network. Detailed formulations and constraints of the analog beamforming matrices for both MiLAC 1 and MiLAC 2 will be presented in the next subsection.

	Let $\mathbf{h}_k \in \mathbb{C}^{L \times 1}$ denote the channel vector between the BS and user $k$, for all $k \in \mathcal{K}$. The channel state information is assumed to be perfectly known at the BS.  The signal received by user-$ k $ is given as
	\begin{equation}
		\begin{split}
			y_k
			=& \mathbf{h}_k^\mathsf{H} \mathbf{x} + n_k, \\
			=& \underbrace{\mathbf{h}_k^\mathsf{H} \mathbf{W}\mathbf{P}^{\frac{1}{2}}  \mathbf{f}_k s_k}_{\text{desired signals}} + \underbrace{\sum_{i \neq k}^K\mathbf{h}_k^\mathsf{H} \mathbf{W}\mathbf{P}^{\frac{1}{2}}  \mathbf{f}_i s_i}_{\text{effective interference}} + \underbrace{n_k}_{\text{effective noise}},
		\end{split}
	\end{equation}
	where $n_k \sim \mathcal{CN}(0, \sigma_k^2)$ denotes the additive white Gaussian noise, and $ \mathbf{f}_k \in \mathbb{C}^{K \times 1} $ is the $ k $-th column of $ \mathbf{F}, \forall k \in \{1, \cdots, K\} $, i.e., $ \mathbf{F}=[\mathbf{f}_1,\cdots,\mathbf{f}_K] $. 
	
	The signal-to-interferenceplus-noise ratio (SINR) for user-$ k $ to decode its desired signal is 
	\begin{equation}\label{eq:SINR}
		\gamma_k = \frac{|\mathbf{h}_k^\mathsf{H}\mathbf{W}\mathbf{P}^{\frac{1}{2}} \mathbf{f}_k|^2}{ \sum_{i\neq k}|\mathbf{h}_k^\mathsf{H}\mathbf{W}\mathbf{P}^{\frac{1}{2}} \mathbf{f}_i|^2 + \sigma_k^2},
	\end{equation}
	where signals from other users are treated as interference. 
	The corresponding achievable rate at user-$ k $ to decode its intended symbol is $R_k = \log ( 1 + 	\gamma_k )$.
	
	\subsection{MiLAC Modeling}
	In this work, we adopt fully-connected MiLAC  architectures \cite{nerini2025milac-capacity} for both MiLAC 1 and MiLAC 2 to explore the optimal performance of the proposed two-layer MiLAC architecture. In the following, we first characterize the model for MiLAC 1, followed by that for MiLAC 2.

	According to microwave theory, the multi-port MiLAC is implemented as a microwave network characterized by a tunable admittance matrix, which depends on the adjustable components of the network. 
	For MiLAC 1, its admittance matrix is denoted by $\mathbf{Y}_1 \in \mathbb{C}^{2K \times 2K}$. Since it is fully-connected,  each port-$v$ of MiLAC 1 is connected to ground through an admittance $Y_{v,v} \in \mathbb{C}$ for all $v \in \mathcal{V} = \{1, \ldots, 2K\}$, and is interconnected to port-$i$ through an admittance $Y_{i,v} \in \mathbb{C}$ for all $i \neq v$. As a function of the tunable admittance components, the $(i,v)$-th entry of the MiLAC admittance matrix $\mathbf{Y}_1$, with $i,v \in \{1, \ldots, 2K\}$, is defined as
	\begin{equation}
		[\mathbf{Y}_1]_{i,v} = 
		\begin{cases}
			- Y_{i,v}, & i \neq v\\[1mm]
			\sum_{j=1}^{2K} Y_{j,v}, & i = v
		\end{cases},
	\end{equation}
	which is obtained by Ohm's law \cite{nerini2025milac1}.

	In the proposed two-layer MiLAC-aided beamforming architecture, the analog beamforming matrix $\mathbf{F}$ associated with MiLAC 1 is expressed as a function of $ \mathbf{Y} \in \mathbb{C}^{2K \times 2K}$ \cite{nerini2025milac2,nerini2025milac-capacity,pozar2021microwave}, i.e.,
	\begin{equation}\label{eq:1-beamforming}
		\mathbf{F} =  \left[\left(\frac{\mathbf{Y}_1}{Y_0}+\mathbf{I}_{2K}\right)^{-1}\right]_{K+1:2K,1:K},
	\end{equation}
	where $ Y_0 $ denotes the reference admittance.
	

	In this work, we assume the MiLACs are lossless and reciprocal. Specifically, the admittance components of a lossless and reciprocal MiLAC are purely imaginary and can be expressed as $Y_{i,v} = j B_{i,v}$, where $B_{i,v} \in \mathbb{R}$ denotes the susceptance of $Y_{i,v}$ for all $i,v \in \{1, \ldots, 2K\}$. Accordingly, a lossless and reciprocal MiLAC at the transmitter must satisfy the following constraints: $\mathbf{Y}_1 = j \mathbf{B}_1, \quad \mathbf{B}_1 = \mathbf{B}_1^\mathsf{T},$
	where $\mathbf{B}_1 \in \mathbb{R}^{N \times N}$ is the susceptance matrix of MiLAC 1.

	This $2K$-port microwave network with admittance matrix $\mathbf{Y}_1 \in \mathbb{C}^{2K \times 2K}$ can be equivalently described by its scattering matrix $\bm \Theta \in \mathbb{C}^{2K \times 2K}$ \cite{nerini2025milac-capacity}, given by $\bm \Theta = (\mathbf{I}_{N} + j Z_0 \mathbf{B}_1)^{-1} (\mathbf{I}_{N} - j Z_0 \mathbf{B}_1)$,
	where $Z_0=Y_0^{-1}$ is the reference impedance used to compute the scattering parameters, typically set to $50\,\Omega$. According to the relationship between scattering matrix and MiLAC-aided beamforming matrix \cite{nerini2025milac-capacity,nerini2025mimo,limits-milac}, \eqref{eq:1-beamforming}  can be reformulated as
	\begin{equation}\label{eq:theta1}
		\mathbf{F} =  \!\left[\left(\frac{\mathbf{Y}_1}{Y_0}+\mathbf{I}_{N}\right)^{-1} \right]_{K+1:2K,1:K}	\!=\frac{1}{2}\left[\bm \Theta\right]_{K+1:2K,1:K}.
	\end{equation}
	
	
	The only difference between MiLAC~1 and MiLAC~2 lies in their dimensionality. Specifically, MiLAC~1 is modeled as a $2K$-port network, whereas MiLAC~2 is a $(L+K)$-port network. With $N = L + K$, define $\mathbf{Y}_2 \in \mathbb{C}^{N \times N}$ as admittance matrix, and let $\bm{\Phi} \in \mathbb{C}^{N \times N}$ denote its scattering matrix of MiLAC~2. Following the modeling of MiLAC 1, we have
	\begin{equation}\label{eq:theta2}
		\begin{split}
			\mathbf{W} = \left[ \left(\frac{\mathbf{Y}_2}{Y_0}+\mathbf{I}_{N}\right)^{-1} \right]_{K+1:N,1:K}
			=\frac{1}{2}\left[\bm \Phi \right]_{K+1:N,1:K}.
		\end{split}
	\end{equation}
	Besides the constraints in \eqref{eq:theta1} and \eqref{eq:theta2}, the lossless and reciprocal MiLAC must adhere to the unitary and symmetry properties, which are given as
	\begin{subequations}
		\begin{align}\label{eq:con-theta1}
			\bm \Theta^\mathsf{H}\bm \Theta =\mathbf{I}_{2K}, \bm \Theta =\bm \Theta^\mathsf{T}, 
		\end{align}
		\begin{align}\label{eq:con-theta2}
			\bm \Phi^\mathsf{H}\bm \Phi=\mathbf{I}_{N}, \bm \Phi=\bm \Phi^\mathsf{T}.
	\end{align}\end{subequations}
	\vspace{-0.2cm}
	\section{Optimality of the Proposed Two-layer MiLAC architecture}
	In this section, we show that, in contrast to single-layer MiLAC-aided beamforming which cannot achieve the full flexibility of digital beamforming, the proposed two-layer MiLAC architecture can attain the same sum-rate performance as fully digital beamforming. The main result of this work is summarized in Theorem \ref{the:MiLAC-HBF} below.

	\begin{theorem}\label{the:MiLAC-HBF}
		Consider a multi-user MISO network with a fully digital beamformer $\mathbf{P}_d \in \mathbb{C}^{L \times K}$ that performs joint magnitude and phase precoding at baseband. Let the singular value decomposition (SVD) of $\mathbf{P}_d$ be:
		\begin{align}\label{eq:Pd-SVD}
			\mathbf{P}_d = \mathbf{U}\mathbf{S}\mathbf{V}^\mathsf{H},
		\end{align}
		where $\mathbf{U} = [\mathbf{U}_1, \mathbf{U}_2] \in \mathbb{C}^{L \times L}$ is a unitary matrix with $\mathbf{U}_1 \in \mathbb{C}^{L \times K}$ and $\mathbf{U}_2 \in \mathbb{C}^{L \times (L-K)}$, 
		$\mathbf{S} \in \mathbb{R}^{L \times K}$ is a diagonal matrix containing the singular values, and $\mathbf{V} \in \mathbb{C}^{K \times K}$ is a unitary matrix. For a fully digital precoding based transmitter with any given precoding matrix $\mathbf{P}_d$, there always exists an equivalent two-layer MiLAC-aided transmitter with optimal scattering matrices $\bm \Theta^\star$ for MiLAC 1, $\bm \Phi^\star$ for MiLAC 2, optimal power allocation matrix $(\mathbf{P}^\star)^{\frac{1}{2}}$, optimal beamforming matrix $\mathbf{F}^\star$ for MiLAC 1 and $\mathbf{W}^\star$ for MiLAC 2 given by 
		\begin{align}\label{eq:theta-star}\small
			\bm\Theta^\star	=
			\begin{bmatrix}
				\mathbf{0}_{K\times K} & \mathbf{V}^\mathsf{*}\\
				\mathbf{V}^\mathsf{H} & \mathbf{0}_{K\times K}
			\end{bmatrix},
		\end{align}
	\vspace{0.1mm}
		\begin{align}\label{eq:phi-star}\small
			\setlength{\jot}{2pt}
			\bm \Phi^\star=
			\begin{bmatrix}
				\mathbf{0}_{K\times K} & \mathbf{U}_1^\mathsf{T}\\
				\mathbf{U}_1 & -\mathbf{U}_2\mathbf{U}_2^\mathsf{T}
			\end{bmatrix},
		\end{align}
		\vspace{-10pt}
		\begin{align}\label{eq:p-star}\small
			(\mathbf{P}^\star)^{\frac{1}{2}} =4\mathbf{S}, 
		\end{align} 
		\vspace{-15pt}
		\begin{align}\small
			\mathbf{F}^\star =\frac{1}{2}\left[\bm \Theta^\star\right]_{K+1:2K,1:K}=\mathbf{V}^\mathsf{H},
		\end{align}
		and
			\vspace{-10pt}
		\begin{equation}\small
			\mathbf{W}^\star =\frac{1}{2}\left[\bm \Phi^\star\right]_{K+1:N,1:K}=\mathbf{U}_1.
		\end{equation}
	The resulting effective MiLAC-based beamforming matrix $ \mathbf{G} \in \mathbb{C}^{L \times K}=\mathbf{W}^\star (\mathbf{P}^\star)^{\frac{1}{2}} \mathbf{F}^\star$
		is equivalent to $\mathbf{P}_d$, showing that the two-layer MiLAC architecture can exactly realize any fully digital beamformer.
	\end{theorem}
		\addtolength{\topmargin}{0.05in}
			\vspace{-25pt}
	\textit{Proof:}  
	For notational convenience, we respectively partition $\bm\Theta \in \mathbb{C}^{2K\times 2K}$ and $\bm \Phi \in \mathbb{C}^{N\times N}$ as 
	$\bm\Theta=
	\begin{bmatrix}
		\bm\Theta_{11} & \bm\Theta_{21}^\mathsf{T}\\
		\bm\Theta_{21} & \bm\Theta_{22}
	\end{bmatrix}$ and $\bm \Phi=
	\begin{bmatrix}
		\bm \Phi_{11}  & \bm \Phi_{21}^\mathsf{T}\\
		\bm \Phi_{21} & \bm \Phi_{22}
	\end{bmatrix}$,
	where $\bm\Theta_{11}\in\mathbb{C}^{K\times K}$, $\bm\Theta_{21}\in\mathbb{C}^{K\times K}$, $\bm\Theta_{22}\in\mathbb{C}^{K\times K}$, $\bm \Phi_{11}\in\mathbb{C}^{K\times K}$, $\bm \Phi_{21}\in\mathbb{C}^{L\times K}$, and $\bm \Phi_{22}\in\mathbb{C}^{L\times L}$. For the two-layer MiLAC architecture proposed in Section \ref{Sec:system model}, we obtain its effective beamformer $ \mathbf{G} \in \mathbb{C}^{L \times K}$ at the transmitter as 
	\begin{equation}\label{eq:eff-G}\small
		\mathbf{G}=\mathbf{W}\mathbf{P}^{\frac{1}{2}}\mathbf{F} = \frac{1}{4}\bm \Phi_{21}\mathbf{P}^{\frac{1}{2}}\bm\Theta_{21}.
	\end{equation}
	The second equality in \eqref{eq:eff-G} is obtained by substituting \eqref{eq:theta1} and \eqref{eq:theta2} into \eqref{eq:eff-G}. To prove Theorem~\ref{the:MiLAC-HBF}, it is sufficient to show that there exists a matrix $\mathbf{G}$ satisfying the constraints of the MiLAC architecture and is equivalent to the fully digital beamformer 
	$\mathbf{P}_d$. This can be established by demonstrating that the following optimization problem always admits a global solution with an objective value of zero:
	\begin{subequations}\label{eq:ls}\small
		\begin{align}
			\min_{\bm\Theta,\bm \Phi,\mathbf{P}^{\frac{1}{2}}} \quad& 
			\left\|\mathbf{P}_d-\frac{1}{4}\bm \Phi_{21}\mathbf{P}^{\frac{1}{2}}\bm\Theta_{21}\right\|_\mathsf{F}^2 \label{eq:ls-obj}\\
			\text{s.t.} \quad 
			&\eqref{eq:con-theta1},\eqref{eq:con-theta2},\label{eq:theta-all}\\
			&\operatorname{tr}(\mathbf{P})\le P_t.\label{eq:power}
		\end{align}
	\end{subequations}
	To solve problem \eqref{eq:ls}, we first consider the SVD of $\mathbf{P}_d$ in \eqref{eq:Pd-SVD}. Since columns of $\mathbf{P}_d$ lie in $\operatorname{span}(\mathbf{U}_1)$, it is natural to choose $\bm \Phi_{21}^\star = \mathbf{U}_1$. For a fixed $\bm \Phi_{21}^\star$,  the objective in \eqref{eq:ls-obj} is minimized by choosing $\bm\Theta_{21}^\star = \mathbf{V}^\mathsf{H}$. 
	However, $\bm \Phi_{21}^\star$ and $\bm\Theta_{21}^\star$ are only sub-blocks of $\bm \Phi^\star$ and $\bm\Theta^\star$, respectively. We must therefore construct full matrices $\bm \Phi^\star$ and $\bm\Theta^\star$ that satisfy the constraints \eqref{eq:con-theta1} and \eqref{eq:con-theta2} in \eqref{eq:ls}. To this end, we begin by verifying the feasibility of $\bm \Phi$, followed by that of $\bm \Theta$. To ensure that the structural constraints in \eqref{eq:theta2} and \eqref{eq:con-theta2} are satisfied, the feasibility of $\bm \Phi$ can be assessed by solving the following problem:
	\begin{subequations}\label{eq:feasible-check}
		\begin{align}
			\text{find}\quad &\bm \Phi_{11},\bm \Phi_{22}\\
			\text{s.t.}\quad &
			\begin{bmatrix}
				\bm \Phi_{11} & \mathbf{U}_1^\mathsf{T}\\
				\mathbf{U}_1 & \bm \Phi_{22}
			\end{bmatrix}^\mathsf{H}
			\begin{bmatrix}
				\bm \Phi_{11} & \mathbf{U}_1^\mathsf{T}\\
				\mathbf{U}_1 & \bm \Phi_{22}
			\end{bmatrix}
			=\mathbf{I}_N, \label{eq:unit-norm}\\
			&
			\bm \Phi_{11}=\bm \Phi_{11}^\mathsf{T},\\
			&
			\bm \Phi_{22}=\bm \Phi_{22}^\mathsf{T}.
		\end{align}
	\end{subequations}

	\vspace{3.5mm}
	Since $\mathbf{U}_1^\mathsf{H}\mathbf{U}_1=\mathbf{I}_K$, the first $K$ columns of $\bm \Phi$ must have unit $\ell_2$ norm. This forces $\bm \Phi_{11}^\star = \mathbf{0}_{K\times K}$.
	The remaining feasibility conditions in \eqref{eq:feasible-check} reduce to
	\begin{subequations}\label{eq:phi22}
		\begin{align}
			\mathbf{U}_1^\mathsf{H} \bm \Phi_{22} &= \mathbf{0}_{K\times L},\\
			\mathbf{U}_1\mathbf{U}_1^\mathsf{H} + \bm \Phi_{22}^\mathsf{H}\bm \Phi_{22} &= \mathbf{I}_L,\\
			\bm \Phi_{22} &= \bm \Phi_{22}^\mathsf{T}.
		\end{align}
	\end{subequations}
	A valid solution for $\bm \Phi_{22}$ is $\bm \Phi_{22}^\star = -\mathbf{U}_2\mathbf{U}_2^\mathsf{T}$,
	which satisfies all three constraints in \eqref{eq:phi22}. Notably, the selection of $\bm \Phi_{22}$ has no impact on system performance. Accordingly, the optimal solution for MiLAC 2 is given by \eqref{eq:phi-star}.
	Following the same procedure as that used to obtain $\bm \Phi^\star$ in \eqref{eq:feasible-check}-\eqref{eq:phi22}, we obtain the optimal solution for the scattering matrix $\bm \Theta$ of MiLAC 1 as \eqref{eq:theta-star}.
	\par
	Now, it remains to optimize $\mathbf{P}^{\frac{1}{2}}$ in \eqref{eq:ls}.
	With $\bm \Phi_{21}^\star=\mathbf{U}_1$ and $\bm\Theta_{21}^\star=\mathbf{V}^\mathsf{H}$, problem \eqref{eq:ls} reduces to
	\begin{subequations}\label{eq:ls-P}
		\begin{align}
			\min_{\mathbf{P}^{1/2}} \quad &
			\big\|\mathbf{S} - \tfrac{1}{4}\mathbf{P}^{\frac{1}{2}}\big\|_\mathsf{F}^2\\
			\text{s.t.}\quad&
			\operatorname{tr}(\mathbf{P})\le P_t.
		\end{align}
	\end{subequations}
	Notice that \eqref{eq:ls-P} is a constrained least-squares problem. The optimal solution can be obtained by projecting $4\mathbf{S}$ onto the Frobenius norm ball, which yields a closed-form expression for $(\mathbf{P}^\star)^{\frac{1}{2}}$. Specifically, the global minimizer of \eqref{eq:ls-P} is given by \eqref{eq:p-star}.

In summary, under the solution $(\bm\Theta^\star,\bm \Phi^\star,(\mathbf{P}^\star)^{\frac{1}{2}})$, the objective value of \eqref{eq:ls} is zero, and this solution satisfies
	\begin{equation}
		\frac{1}{4}\bm \Phi^\star_{21}(\mathbf{P}^\star)^{\frac{1}{2}}\bm\Theta^\star_{21}
		= \mathbf{U}_1\mathbf{S}\mathbf{V}^\mathsf{H}
		= \mathbf{P}_d,
	\end{equation}
	which proves that two-layer MiLAC-aided beamforming can always reproduce the fully digital beamformer when both MiLAC 1 and MiLAC 2 are lossless and reciprocal. Therefore, we complete the proof.
	\hfill$\square$

	\section{Problem Formulation and Proposed Algorithm}
	\subsection{Problem Formulation}
	Based on the proposed two-layer MiLAC transmission network, we next jointly optimize the precoders of the two MiLACs $\mathbf{F}, \mathbf{W}$ and the amplifier power allocation $\mathbf{P}$ to maximize system sum-rate. This is equivalent to optimizing the scattering matrices $\bm \Theta$ and $ \bm \Phi$ of MiLAC 1 and MiLAC 2, along with $\mathbf{P}$. The resulting sum-rate maximization problem is formulated as follows:
	\begin{subequations}\label{eq:pro-formula}
		\begin{align}
			\max_{\mathbf{P}, \bm \Theta, \bm \Phi} \quad &\sum_{k=1}^K R_k\\
			\operatorname{s.t.}  \quad 	& \eqref{eq:theta1}, \eqref{eq:theta2}, \eqref{eq:con-theta1}, \eqref{eq:con-theta2}, \eqref{eq:power},
		\end{align}
	\end{subequations}
	where constraints \eqref{eq:theta1}, \eqref{eq:theta2}, \eqref{eq:con-theta1}, \eqref{eq:con-theta2} ensure that the analog beamformers of both MiLACs are realized by lossless and reciprocal microwave networks, and \eqref{eq:power} is the power allocation constraint. 
	Problem \eqref{eq:pro-formula} is highly non-convex and challenging to solve due to the coupling between $ \mathbf{F} $ and $ \mathbf{W} $ in fractional rate expressions, along with their unitary and symmetric constraints. The algorithm we proposed to solve \eqref{eq:pro-formula} is specified in the next subsection.
	\subsection{Proposed Algorithm}
	Inspired by Theorem \ref{the:MiLAC-HBF}, we propose to solve \eqref{eq:pro-formula} in the following two steps:
	\begin{itemize}
		\item \textit{Step 1:} Solve the sum-rate problem for a conventional fully digital beamforming system and find optimal digital beamformer $\mathbf{P}_d^\star$.
		\item \textit{Step 2:} With $\mathbf{P}_d^\star$ obtained in Step 1, apply Theorem \ref{the:MiLAC-HBF} to derive closed-form solutions for the optimal MiLAC scattering matrices $\bm \Theta^\star$, $\bm \Phi^\star$, and power allocation $\mathbf{P}^\star$.
	\end{itemize}
	
	Since the closed-form mapping required in Step 2 is provided by Theorem \ref{the:MiLAC-HBF}, the main challenge reduces to finding the optimal fully digital beamforming in Step 1. For a conventional multi-user MISO network with $L$ RF chains, the sum-rate maximization problem that optimizes the fully digital beamforming matrix $\mathbf{P}_d \in \mathbb{C}^{L \times K}$ is given by:
	\begin{subequations}\label{eq:pro-formula-di}
		\begin{align}
			\max_{\mathbf{P}_d} \quad &\sum_{k=1}^K R^d_k\\
			\operatorname{s.t.}  \quad 	&   \operatorname{tr}(\mathbf{P}_d\mathbf{P}_d^\mathsf{H}) = P_t, \label{eq:power-digital}
		\end{align}
	\end{subequations}
	where $ \mathbf{P}_d=[\mathbf{p}_1^d, \cdots, \mathbf{p}_K^d] $ and
	\begin{equation}\label{eq:rate-di}
		R_k^d = \log \left( 1 + 	\frac{|\mathbf{h}_k^H\mathbf{p}_k^d |^2}{ \sum_{i\neq k}|\mathbf{h}_k^H \mathbf{p}^d_i|^2 + \sigma_k^2} \right).
	\end{equation}
	Problem \eqref{eq:pro-formula-di} has been widely investigated in the literature \cite{WMMSE2011,rethingkingWMMSE}. Typical methods to address \eqref{eq:pro-formula-di} includes conventional weighted minimum mean-square error (WMMSE) algorithm \cite{WMMSE2011} or its reduced-complexity variant, reduced WMMSE (RWMMSE) \cite{rethingkingWMMSE}. However, both approaches require matrix inversion, leading to prohibitive computational complexity in massive MIMO systems, with complexity scaling cubically or linearly in the number of transmit antennas. To overcome these limitations, we next propose an efficient algorithm that significantly reduces computational complexity while preserving performance.

		The channel matrix from the BS to all users is denoted as $ \mathbf{H} = [\mathbf{h}_1, \cdots, \mathbf{h}_K] \in \mathbb{C}^{L \times K} $. It has been proved in \cite{universal_low2025,fang2023multi-group} that any optimal solution $\mathbf{P}_d^\star$ of problem \eqref{eq:pro-formula-di} lies within the range space of the full channel matrix $\mathbf{H}$, i.e.,
	\begin{equation}\label{eq:OBS-weight}
		\mathbf{P}_d^\star = \mathbf{H} \mathbf{M},
	\end{equation}
	where $\mathbf{M} \in \mathbb{C}^{K \times K}$ is an coefficient matrix. This indicates that each optimal beamforming vector $(\mathbf{p}_{k}^d)^\star$ lies in the column space of the  channel matrix $\mathbf{H}$. 
	Based on this structure, we then propose a novel method that 
	reduces the dimensionality of the optimization variables.
	To further exploit the structure in \eqref{eq:OBS-weight}, consider the SVD of $\mathbf{H}$, given by $\mathbf{H} = \mathbf{Q} \bm{\Sigma} \mathbf{R}^\mathsf{H}$, where $\mathbf{Q} \in \mathbb{C}^{L \times K}$, $\bm{\Sigma} \in \mathbb{C}^{K \times K}$, and $\mathbf{R} \in \mathbb{C}^{K \times K}$. Based on this decomposition, we define the matrix $\mathbf{T} = \bm{\Sigma} \mathbf{R}^\mathsf{H} \mathbf{M}\in \mathbb{C}^{K \times K}$. Hence, $\mathbf{P}_d$ is equivalent to $\mathbf{Q}\mathbf{T}$,  and optimizing $\mathbf{P}_d$ in \eqref{eq:pro-formula-di} reduces to optimizing $\mathbf{T}$, whose dimension is independent of the number of transmit antennas $L$. This simplifies the optimization problem to: 
	\begin{subequations}\label{eq:pro-reformula-di}\small
		\begin{align}
			\max_{\mathbf{T}} \quad &\sum_{k=1}^K R_k\\
			\operatorname{s.t.}  \quad 	
			& \operatorname{tr}(\mathbf{T} \mathbf{T}^\mathsf{H}) = P_t,\label{eq:power-constraint-re}
		\end{align}
	\end{subequations} 	
	where $ R_k = \log (1+\frac{|\bar{\mathbf{h}}_k^\mathsf{H}\mathbf{t}_k|^2}{ \sum_{i\neq k}|\bar{\mathbf{h}}_k^\mathsf{H}\mathbf{t}_i|^2 + \sigma_k^2}) $,
	$ \bar{\mathbf{h}}_k= \mathbf{Q}^\mathsf{H} \mathbf{h}_k \in \mathbb{C}^{K\times 1}$ and $ \mathbf{t}_k \in \mathbb{C}^{K\times 1}$ is the $ k $-th column of $ \mathbf{T} $, i.e., $ \mathbf{T} =[\mathbf{t}_1, \cdots, \mathbf{t}_K] $. 
	
	Next, inspired by \cite{202501joint}, we address \eqref{eq:pro-reformula-di} using the projected successive linear approximation (PSLA)-based framework. First, the fractional programming (FP) technique \cite{shen2018fractional} is employed to transform \eqref{eq:pro-reformula-di} into a block-convex problem via auxiliary variables. Then, an alternating optimization (AO) algorithm is developed, where each auxiliary block admits a closed-form solution and the $\mathbf{T}$-block is updated via a projection-based closed-form expression.
	
	Specifically, introducing auxiliary variables $ \alpha_k $ and  $ \beta_k $, the expression $R_k$ in \eqref{eq:pro-reformula-di} is reformulated as
	\begin{equation}\label{eq:rate-beta}
		\begin{split}
			\bar{R}_k =& 2\sqrt{1+\alpha_k}\Re\left\{\beta_k^*\bar{\mathbf{h}}_k^\mathsf{H} \mathbf{t}_k\right\}+ \log(1+\alpha_k)\\
			&-|\beta_k|^2\left( \sum_{i=1}^{K} |\bar{\mathbf{h}}_k^\mathsf{H} \mathbf{t}_i|^2 +\sigma_k^2 \right) -\alpha_k.
		\end{split}
	\end{equation}
	By further introducing the following matrix variables $\bm \Sigma_1 =\operatorname{diag}\left(\sqrt{(1+\alpha_1)}\beta_1, \cdots, \sqrt{(1+\alpha_K)}\beta_K\right), \bm \Sigma_2 =\operatorname{diag}\left(|\beta_1|^2, \cdots, |\beta_K|^2\right), \bar{\mathbf{H}}=\left[\bar{\mathbf{h}}_1, \cdots, \bar{\mathbf{h}}_K\right], \bm \alpha=\{\alpha_1,\cdots,\alpha_K\}, \bm \beta=\{\beta_1,\cdots,\beta_K\}$,
	the objective in problem \eqref{eq:pro-reformula-di} is reformulated as a more compact form as follows
	\begin{equation}\label{eq:block-con}
			G(\bm \alpha, \bm \beta,  \mathbf{T})\!=2\Re\!\left\{\operatorname{tr}\left( \bm \Sigma_1\bar{\mathbf{H}}^\mathsf{H}\mathbf{T}\right)\!\right\}\!-\!\operatorname{tr}\left(\mathbf{T}\mathbf{T}^\mathsf{H}  \bar{\mathbf{H}}\bm \Sigma_2 \bar{\mathbf{H}}^\mathsf{H}\right).
	\end{equation}
Notice that the problem, with \eqref{eq:block-con} as the objective function and \eqref{eq:power-constraint-re} as the constraint, is block-wise convex. We therefore adopt an AO algorithm, decomposing it into block $ \{\bm \alpha, \bm \beta\}$ and block $ \{\mathbf{T}\} $. In the following, we derive the closed-form solutions for $ \{\bm \alpha, \bm \beta\}, \{\mathbf{T}\} $, respectively.
	\subsubsection{Subproblem for $\bm{\alpha}$ and $\bm{\beta}$}
	Since both subproblems with respect to $\bm{\alpha}$ and $\bm{\beta}$ are convex and unconstrained, their optimal solutions admit closed forms:
	\begin{equation}\label{eq:alpha}
		\alpha_k^\star=\frac{|\bar{\mathbf{h}}_k^\mathsf{H}\mathbf{t}_k|^2}{ \sum_{i\neq k}|\bar{\mathbf{h}}_k^\mathsf{H} \mathbf{t}_i|^2 + \sigma_k^2 },\;
		\beta_k^\star=\frac{\sqrt{1+\alpha_k}\bar{\mathbf{h}}_k^\mathsf{H} \mathbf{t}_k}{ \sum_{i=1}^{K} |\bar{\mathbf{h}}_k^\mathsf{H} \mathbf{t}_i|^2 + \sigma_k^2 }.
	\end{equation}

	
	\subsubsection{Subproblem for $ \mathbf{T} $}    
	With $ \bm \alpha $ and $ \bm \beta $ fixed, the subproblem with respect to $ \mathbf{T} $ is given by
	\begin{equation}\label{eq:pro-reformula-P}
			\max_{ \mathbf{T}} \quad  G\left( \mathbf{T}\right), \quad 
			\operatorname{s.t.}  \quad
			 \eqref{eq:power-constraint-re}.
	\end{equation}

	%

		Define the feasible set $ \mathcal{P} = \{\mathbf{T}\in \mathbb{C}^{K \times K}|\operatorname{tr}(\mathbf{T}\mathbf{T}^\mathsf{H})=P_t\} $. By introducing a constant factor $\xi P_t$ into the objective, problem \eqref{eq:pro-reformula-P} can be equivalently reformulated as
	\begin{equation}\label{eq:power-shift}\small
		\begin{split}
			\max_{\mathbf{T} \in \mathcal{P}} \; 2\Re\left\{\operatorname{tr}\left( \bm \Sigma_1\bar{\mathbf{H}}^\mathsf{H}\mathbf{T}\right)\!\right\}+\operatorname{tr}\!\left(\mathbf{T}\mathbf{T}^\mathsf{H}\left(\xi\mathbf{I}_K-\bar{\mathbf{H}}\bm \Sigma_2 \bar{\mathbf{H}}^\mathsf{H}\right)\right).
		\end{split}
	\end{equation}

	By further introducing an auxiliary matrix $\overline{\mathbf{T}} \in \mathbb{R}^{K\times K}$, we can transform \eqref{eq:power-shift} into
	\begin{equation}\label{eq:power-linear}\small
		\begin{split}
			\max_{\mathbf{T},\overline{\mathbf{T}}} &\quad 2\Re\left\{\operatorname{tr}\left( \bm \Sigma_1\bar{\mathbf{H}}^\mathsf{H}\mathbf{T}\right)\right\}-\operatorname{tr}\left(\overline{\mathbf{T}}\overline{\mathbf{T}}^\mathsf{H}\left(\xi\mathbf{I}_K-\bar{\mathbf{H}}\bm \Sigma_2 \bar{\mathbf{H}}^\mathsf{H}\right)\right)\\ &+2\Re \left\{\operatorname{tr}\left(\overline{\mathbf{T}}\mathbf{T}^\mathsf{H}\left(\xi\mathbf{I}_K-\bar{\mathbf{H}}\bm \Sigma_2 \bar{\mathbf{H}}^\mathsf{H}\right)\right)\right\},
		\end{split}
	\end{equation}
	where the optimal $ \overline{\mathbf{T}}^\star = \mathbf{T}$.

Let $ \bm\Pi_{\mathcal P}(\cdot):\mathbb{C}^{K\times K} \longrightarrow \mathcal P $ denote the projection operator onto $\mathcal{P}$. With a given $ \overline{\mathbf{T}} $, the optimal $ {\mathbf{T}}^\star $ of problem \eqref{eq:pro-reformula-P} is
	\begin{equation}\label{eq:update-F}
		{\mathbf{T}}^\star = \bm \Pi_{\mathcal{P}}	\left(\bar{\mathbf{H}}\bm \Sigma_1+\left(\lambda_1\mathbf{I}_K-  \bar{\mathbf{H}}\bm \Sigma_2 \bar{\mathbf{H}}^\mathsf{H}\right)\overline{\mathbf{T}}\right),
	\end{equation}
	where  $	\bm \Pi_{\mathcal{P}}(\mathbf{T}) = \arg\min_{\mathbf{Z} \in \mathcal{P}}\|\mathbf{T}-\mathbf{Z}\|^2_\mathsf{F}=\sqrt{\frac{P_t}{\operatorname{tr}(\mathbf{T}\mathbf{T}^\mathsf{H})}}\mathbf{T}$.

	\subsection{Computational Complexity Analysis}
	The computational complexity of the proposed PSLA algorithm is dominated by the update of $\mathbf{T}$ in \eqref{eq:update-F}, which requires $\mathcal{O}(K^3)$ operations. In addition, the SVD in Theorem \ref{the:MiLAC-HBF} requires $\mathcal{O}(L K^2)$ operations. Therefore, the overall computational complexity is $\mathcal{O}(I_1 I_2 K^3 + L K^2)$, where $I_1$ and $I_2$ denote the numbers of iterative update for $\{\bm{\alpha}, \bm{\beta}\}$ and $\mathbf{T}$, respectively.

	\section{Simulation Results}
	In this section, we provide numerical results to validate the theoretical analysis and to demonstrate the effectiveness of the proposed architecture and algorithm.
	
		\begin{figure}[t!]
		\centering
		
		\subfloat[\label{fig:convergence}Convergence.]
		{\includegraphics[width=0.31\textwidth]{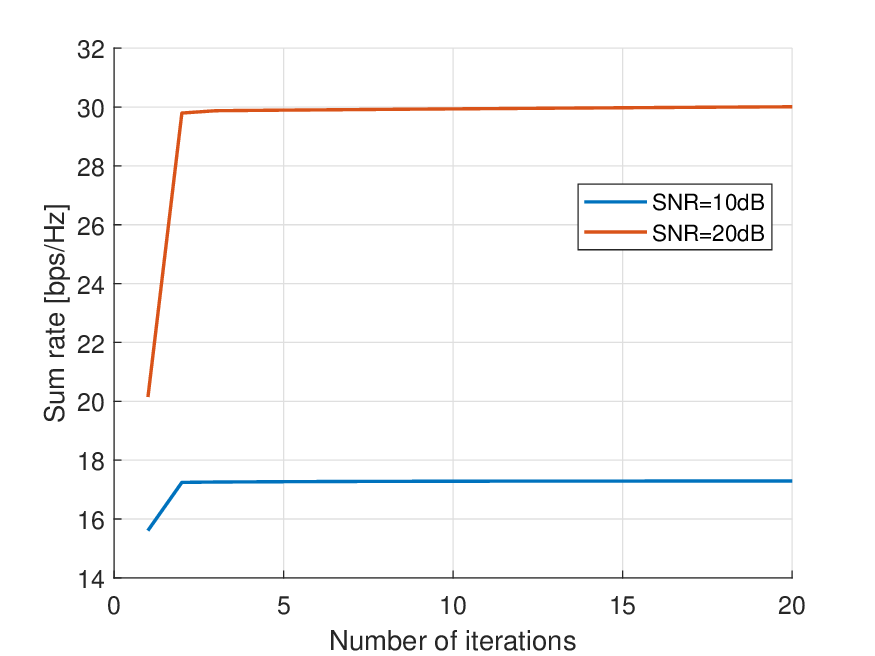}}
		\hfill
		\subfloat[\label{fig:differ}Algorithm comparison.]
		{\includegraphics[width=0.31\textwidth]{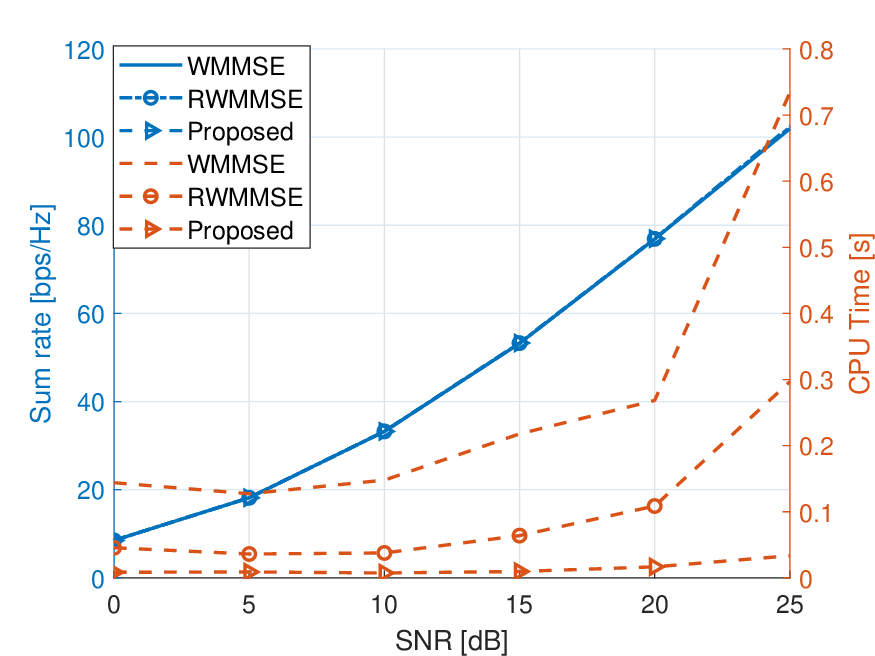}}
		\caption{Performance of the proposed PSLA-based algorithm.}
	\end{figure}
	
	\begin{figure}[t!]
		\centering
		\subfloat[\label{fig:WSR-Pt} $ K=4 $.]
		{\includegraphics[width=0.35\textwidth]{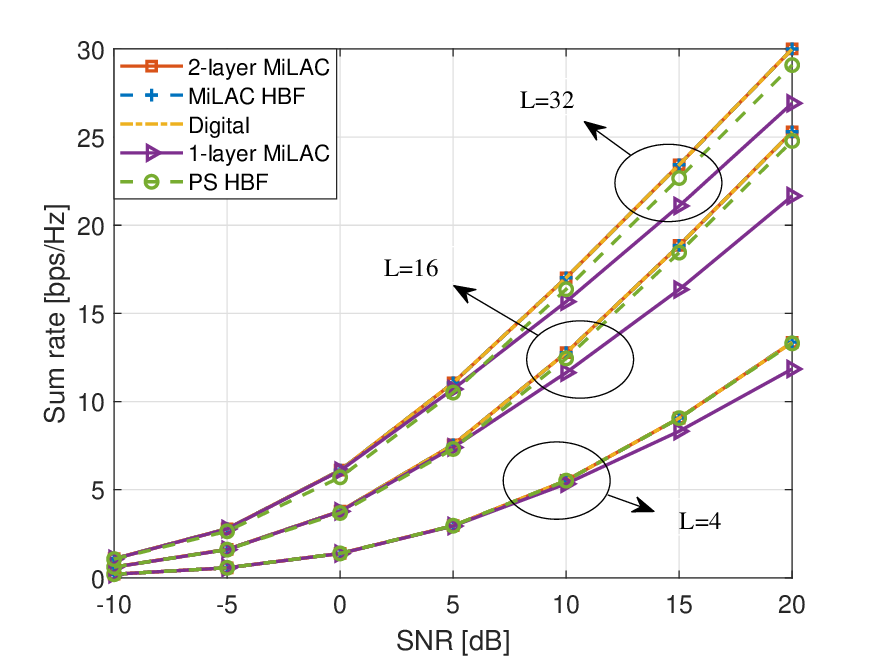}}
		\hfill
		\subfloat[\label{fig:WSR-L} $ K=8 $.]
		{\includegraphics[width=0.35\textwidth]{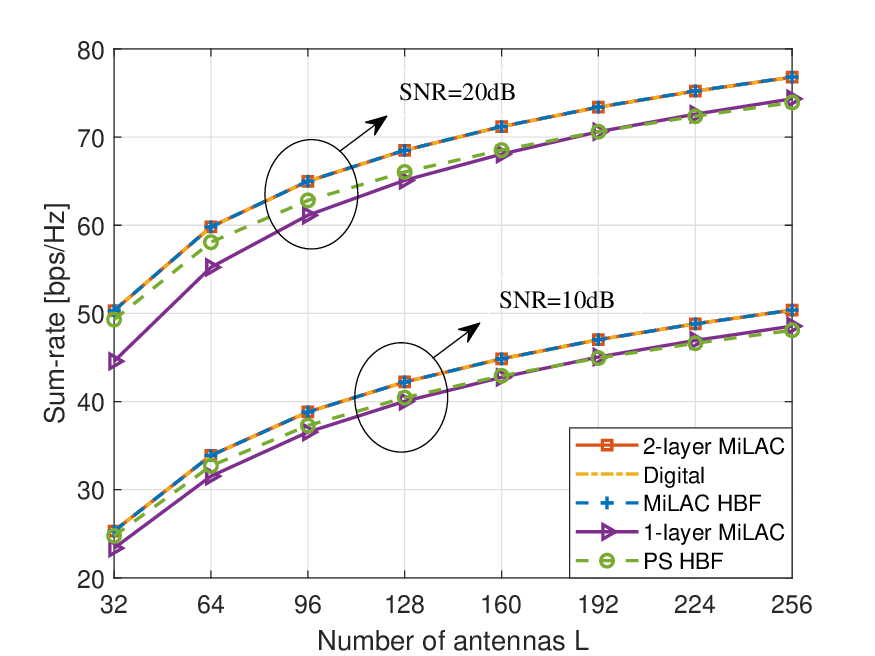}}
		
		\caption{Comparison of different beamforming architectures.}
	\end{figure}	
	
	The channel of user-$k$ is generated i.i.d. as $\mathbf{h}_k \sim (\mathbf{0}, \mathbf{I}_L)$, and the noise variance is fixed at $\sigma_k=1, \forall k \in \mathcal{K}$. The convergence tolerance is $\epsilon=10^{-4}$. All simulation results are averaged over 100 random channel realizations. The following beamforming architectures are considered for comparison:
	(1) Digital: a fully digital architecture with $L$ RF chains and an $L \times K$ beamformer;
	(2) 1-layer MiLAC \cite{nerini2025milac-capacity}: a $(L+K)$-port MiLAC-aided analog architecture with $K$ RF chains and an $L \times K$ beamformer;
	(3) MiLAC HBF \cite{limits-milac}: a hybrid architecture with $K$ RF chains, combining a $K \times K$ digital beamformer and an $L \times K$ MiLAC-based analog beamformer;
	(4) PS HBF \cite{yu2016alternating}: a fully connected hybrid architecture with $K$ RF chains and phase shifter (PS)-based analog beamforming;
	(5) 2-layer MiLAC: the proposed fully analog architecture.

	Fig. \ref{fig:convergence} illustrates the convergence behavior of the proposed algorithm under different transmit signal-to-noise ratio (SNR) values, with $L=32$ and $K=4$. It can be observed that the proposed algorithm converges within a few iterations. Fig. \ref{fig:differ} demonstrates that the proposed algorithm achieves comparable performance to the baseline methods (i.e., WMMSE and RWMMSE) while significantly reducing the average CPU time. For a fair comparison, the PSLA-based algorithm is applied to all architectures in this section.
	
	Fig.~\ref{fig:WSR-Pt} illustrates the sum-rate performance of the proposed \textit{2-layer MiLAC} as a function of the transmit SNR under different antenna configurations and user settings, in comparison with four conventional architectures: \textit{digital}, \textit{1-layer MiLAC}, \textit{MiLAC HBF}, and \textit{PS HBF}.  It can be observed that \textit{digital}, \textit{MiLAC HBF}, and \textit{2-layer MiLAC} achieve identical performance, and all exhibit increasing sum-rate performance as the SNR increases. The equivalent performance of the proposed \textit{2-layer MiLAC} and fully digital beamforming architectures  perfectly validates the theoretical analysis specified in Theorem \ref{the:MiLAC-HBF}.
	
	Fig.~\ref{fig:WSR-L} illustrates the sum-rate performance as the number of transmit antennas increases for $K=8$. It can be observed that the performance of \textit{1-layer MiLAC} gradually approaches that of \textit{digital} as the number of antennas increases, whereas the performance of \textit{PS HBF} is worse than that of the \textit{digital} approach, especially as the number of antennas increases. In contrast, the proposed \textit{2-layer MiLAC} achieves the same performance as the fully digital scheme while requiring only $K$ low-resolution RF chains, whereas the digital architecture relies on $L$ RF high-resolution RF chains. This demonstrates the generality and practical effectiveness of the proposed architecture.
	\vspace{-0.4cm}
	\section{Conclusion}
	This work proposed a novel beamforming architecture, referred to as the two-layer MiLAC-aided beamforming architecture. Building upon this design, we theoretically demonstrated that the proposed fully analog architecture can achieve performance equivalent to fully digital beamforming. Furthermore, we formulated a sum-rate maximization problem and developed an efficient algorithm for its solution. Simulation results validate the theoretical analysis and show the effectiveness of the proposed two-layer MiLAC-aided beamforming architecture, highlighting its potential as a competitive solution for future large-scale MIMO systems.
	

	
	\bibliographystyle{IEEEtran}  
	\bibliography{reference}

\begin{thebibliography}{10}
\providecommand{\url}[1]{#1}
\csname url@samestyle\endcsname
\providecommand{\newblock}{\relax}
\providecommand{\bibinfo}[2]{#2}
\providecommand{\BIBentrySTDinterwordspacing}{\spaceskip=0pt\relax}
\providecommand{\BIBentryALTinterwordstretchfactor}{4}
\providecommand{\BIBentryALTinterwordspacing}{\spaceskip=\fontdimen2\font plus
\BIBentryALTinterwordstretchfactor\fontdimen3\font minus
  \fontdimen4\font\relax}
\providecommand{\BIBforeignlanguage}[2]{{%
\expandafter\ifx\csname l@#1\endcsname\relax
\typeout{** WARNING: IEEEtran.bst: No hyphenation pattern has been}%
\typeout{** loaded for the language `#1'. Using the pattern for}%
\typeout{** the default language instead.}%
\else
\language=\csname l@#1\endcsname
\fi
#2}}
\providecommand{\BIBdecl}{\relax}
\BIBdecl

\bibitem{BDRIStutorial}
H.~Li, M.~Nerini, S.~Shen, and B.~Clerckx, ``A tutorial on beyond-diagonal
  reconfigurable intelligent surfaces: Modeling, architectures, system design
  and optimization, and applications,'' \emph{{IEEE} Commun. Surveys Tuts.},
  2025.

\bibitem{nerini2025milac1}
M.~Nerini and B.~Clerckx, ``Analog computing for signal processing and
  communications — {P}art {I}: Computing with microwave networks,''
  \emph{IEEE Trans. Signal Process.}, pp. 1--15, Dec. 2025.

\bibitem{nerini2025milac2}
------, ``Analog computing for signal processing and communications — {P}art
  {II}: Toward gigantic {MIMO} beamforming,'' \emph{IEEE Trans. Signal
  Process.}, pp. 1--15, Dec. 2025.

\bibitem{nerini2025milac-capacity}
------, ``Capacity of {MIMO} systems aided by microwave linear analog computers
  ({MiLACs}),'' \emph{arXiv preprint arXiv:2506.05983}, 2025.

\bibitem{nerini2025mimo}
------, ``{MIMO} systems aided by microwave linear analog computers:
  Capacity-achieving architectures with reduced circuit complexity,''
  \emph{{IEEE} Trans. Wireless Commun.}, Mar. 2026.

\bibitem{multiuser-milac}
T.~Fang, X.~Zhou, and Y.~Mao, ``On the performance of lossless reciprocal
  {MiLAC} architectures in multi-user networks,'' \emph{{IEEE} Wireless Commun.
  Lett.}, pp. 1--1, Apr. 2026.

\bibitem{limits-milac}
Z.~Wu, M.~Nerini, and B.~Clerckx, ``Microwave linear analog computer
  ({MiLAC})-aided multiuser {MISO}: Fundamental limits and beamforming
  design,'' \emph{arXiv preprint arXiv:2601.10060}, 2026.

\bibitem{pozar2021microwave}
D.~M. Pozar, \emph{Microwave engineering: theory and techniques}.\hskip 1em
  plus 0.5em minus 0.4em\relax John wiley \& sons, 2021.

\bibitem{WMMSE2011}
Q.~Shi, M.~Razaviyayn, Z.-Q. Luo, and C.~He, ``An iteratively weighted {MMSE}
  approach to distributed sum-utility maximization for a {MIMO} interfering
  broadcast channel,'' \emph{IEEE Trans. Signal Process.}, vol.~59, no.~9, pp.
  4331--4340, Sep. 2011.

\bibitem{rethingkingWMMSE}
X.~Zhao, S.~Lu, Q.~Shi, and Z.-Q. Luo, ``Rethinking {WMMSE}: Can its complexity
  scale linearly with the number of bs antennas?'' \emph{IEEE Trans. Signal
  Process.}, vol.~71, pp. 433--446, Feb. 2023.

\bibitem{universal_low2025}
X.~Zhao and Q.~Shi, ``A universal low-dimensional subspace structure in
  beamforming design: Theory and applications,'' \emph{IEEE Trans. Signal
  Process.}, vol.~73, pp. 1775--1791, Apr. 2025.

\bibitem{fang2023multi-group}
T.~Fang and Y.~Mao, ``Optimal beamforming structure and efficient optimization
  algorithms for generalized multi-group multicast beamforming optimization,''
  \emph{{IEEE} Trans. Signal Process.}, pp. 1--16, Jun. 2025.

\bibitem{202501joint}
X.~Zhou, T.~Fang, and Y.~Mao, ``Joint active and passive beamforming
  optimization for beyond diagonal {RIS}-aided multi-user communications,''
  \emph{IEEE Commun. Lett.}, pp. 1--1, Jan. 2025.

\bibitem{shen2018fractional}
K.~Shen and W.~Yu, ``Fractional programming for communication systems—part
  {I}: Power control and beamforming,'' \emph{IEEE Trans. Signal Process.},
  vol.~66, no.~10, pp. 2616--2630, May. 2018.

\bibitem{yu2016alternating}
X.~Yu, J.-C. Shen, J.~Zhang, and K.~B. Letaief, ``Alternating minimization
  algorithms for hybrid precoding in millimeter wave {MIMO} systems,''
  \emph{IEEE J. Sel. Topics Signal Process.}, vol.~10, no.~3, pp. 485--500,
  2016.

\end{thebibliography}

\end{document}